\RequirePackage{fix-cm}
\documentclass[smallextended]{svjour3}       
\smartqed  
\usepackage{graphicx}
\usepackage{mathptmx}      
\usepackage{amsmath,amssymb}
\usepackage{enumerate}
\usepackage{braket}
\usepackage{bbm}
\usepackage{dsfont}

\newcommand{\defarrow}{\stackrel{\mathrm{def.}}{\Leftrightarrow}}

\newcommand{\cmplx}{\mathbb{C}}

\newcommand{\tr}{\mathrm{tr}}

\newcommand{\calL}{\mathcal{L}}
\newcommand{\cH}{\mathcal{H}}

\newcommand{\cK}{\mathcal{K}}
\newcommand{\LH}{\mathcal{L} (\mathcal{H}) }

\newcommand{\LK}{\mathcal{L}  (\mathcal{K})}

\newcommand{\s}{\mathrm{s}}
\newcommand{\ts}{\tilde{\mathrm{s}}}

\newcommand{\cstar}{$C^\ast$}
\newcommand{\A}{\mathcal{A}}
\newcommand{\B}{\mathcal{B}}
\newcommand{\C}{\mathcal{C}}
\newcommand{\D}{\mathcal{D}}

\newcommand{\Mn}{\mathbb{M}_n}

\newcommand{\M}{\mathcal{M}}
\newcommand{\N}{\mathcal{N}}

\newcommand{\vph}{\varphi}

\newcommand{\Ss}{\mathcal{S}_{\sigma}}

\newcommand{\zent}{\mathfrak{Z}}

\newcommand{\cpchset}[2]{\mathbf{Ch}^{\mathrm{CP}} (#1 \to #2)}

\newcommand{\ncpchset}[2]{\mathbf{Ch}^{\mathrm{CP}}_\sigma (#1 \to #2)}
\newcommand{\ncpch}[1]{\mathbf{Ch}^{\mathrm{CP}}_\sigma  (#1)}
\newcommand{\schset}[2]{\mathbf{Ch}^{\mathrm{Sch}} (#1 \to #2)}

\newcommand{\nschset}[2]{\mathbf{Ch}^{\mathrm{Sch}}_\sigma (#1 \to #2)}
\newcommand{\nsch}[1]{\mathbf{Ch}^{\mathrm{Sch}}_\sigma  (#1)}

\newcommand{\atensor}{\otimes_{\mathrm{alg}}}
\newcommand{\ntensor}{\overline{\otimes}}
\newcommand{\maxtensor}{\otimes_{\mathrm{max}}}
\newcommand{\mintensor}{\otimes_{\mathrm{min}}}
\newcommand{\gtensor}{\otimes_{\gamma}}

\newcommand{\id}{\mathrm{id}}
\newcommand{\unit}{\mathds{1}}

\newcommand{\cocp}{\preccurlyeq_{\mathrm{CP}}}
\newcommand{\eqcp}{\sim_{\mathrm{CP}}}
\newcommand{\concp}{\preccurlyeq_{ \mathrm{CP}_\sigma}}
\newcommand{\eqncp}{\sim_{\mathrm{CP}_\sigma}}

\newcommand{\norm}[1]{\lVert #1 \rVert}
\newcommand{\maxnorm}[1]{\norm{ #1 }_{\mathrm{max}}}
\newcommand{\minnorm}[1]{\norm{ #1 }_{\mathrm{min}}}

\newcommand{\gnorm}[1]{\norm{ #1 }_{\gamma}}

\newcommand{\rth}{\rho_\theta}
\newcommand{\phth}{\varphi_\theta}
\newcommand{\psth}{\psi_\theta}
\newcommand{\psx}{\psi_\xi}
\newcommand{\E}{\mathcal{E}}
\newcommand{\F}{\mathcal{F}}
\newcommand{\seE}{(\M , \Theta , (\varphi_\theta)_{\theta \in \Theta} )}
\newcommand{\seEz}{(\M_0 , \Theta , (\varphi^{(0)}_\theta)_{\theta \in \Theta} )}
\newcommand{\seF}{(\N , \Theta , (\psi_\theta)_{\theta \in \Theta} )}
\newcommand{\seFx}{(\N , \Xi , (\psi_\xi)_{\xi \in \Xi} )}
\newcommand{\seFxz}{(\N_0 , \Xi , (\psi_\xi^{(0)})_{\xi \in \Xi} )}
\newcommand{\phthz}{\varphi_\theta^{(0)}}
\newcommand{\phtho}{\varphi_\theta^{(1)}}

\newcommand{\pstho}{\psi_\theta^{(1)}}

\newcommand{\thin}{\theta \in \Theta}
\newcommand{\xin}{\xi \in \Xi}

\newcommand{\cl}{\mathrm{cl}}
\newcommand{\clp}{\Subset}

\newcommand{\CP}{\mathrm{CP}} 
\newcommand{\Sch}{\mathrm{Sch}} 
\newcommand{\alg}{\mathrm{alg}} 
\newcommand{\bin}{\mathrm{bin}} 
\newcommand{\lnor}{\mathrm{lnor}}

\newcommand{\condi}{\mathbb{E}}

\begin{document}

\title{Accessible information without disturbing partially known quantum states on 
a von Neumann algebra
\thanks{
This work was supported
by the National Natural Science Foundation of China (Grants No.~11374375 and
No.~11574405).
}
}

\titlerunning{Accessible information without disturbing partially known quantum states}        

\author{
Yui Kuramochi
}


\institute{Yui Kuramochi \at
School of Physics and Astronomy, Sun Yat-Sen University (Zhuhai Campus),
Zhuhai 519082, China\\
\email{yui.tasuke.kuramochi@gmail.com}           
}

\date{Received: date / Accepted: date}

\maketitle

\begin{abstract}
This paper addresses the problem of how much information we can extract
without disturbing a statistical experiment, 
which is a family of 
partially known normal states on a von Neumann algebra.
We define the classical part of a statistical experiment
as the restriction of the equivalent minimal sufficient statistical experiment
to the center of the outcome space,
which, in the case of density operators on a Hilbert space,
corresponds to the classical probability distributions
appearing in the maximal decomposition
by Koashi and Imoto
[Phys.~Rev.~A \textbf{66}, 022318 (2002)].
We show that we can access by a Schwarz or completely positive channel
at most the classical part of a statistical experiment
if we do not disturb the states.
We apply this result to the broadcasting problem of a statistical experiment.
We also show that the classical part of the direct product of 
statistical experiments is the direct product of the classical parts of 
the statistical experiments.
The proof of the latter result is based on the theorem
that the direct product of minimal sufficient statistical experiments 
is also minimal sufficient.
\end{abstract}

\section{Introduction}
\label{sec:intro}

One of the fundamental feature of quantum theory 
is the impossibility of extracting information without disturbing 
unknown quantum states.
We can see this characteristic, for example,
from the no-cloning theorem~\cite{Park1970,wootters1982single,DIEKS1982271}
or 
more general no-broadcasting theorem~\cite{PhysRevLett.76.2818,PhysRevLett.99.240501},
which states that cloning or broadcasting operation can be realized
only for restricted family of quantum states.

Recently, the cloning and broadcasting conditions were considered 
in general operator algebraic framework~\cite{Kaniowski2015,kaniowski_2017},
in which the mean ergodic theorem 
for von Neumann algebras~\cite{ISI:A1979HD47500016}
plays a fundamental role. 
With the help of the mean ergodic theorem,
we can also establish the existence of 
a minimal sufficient statistical experiment 
equivalent to a given 
operator algebraic statistical experiment \cite{kuramochi2017minimal}.
If a statistical experiment is a family of density operators on a Hilbert space, 
the equivalent minimal sufficient statistical experiment corresponds to
the maximal decomposition by 
Koashi and Imoto~\cite{PhysRevA.66.022318}.

We can regard the broadcasting as a special class of 
operations that do not disturb a family of partially known states.
A typical example of such an operation 
other than the broadcasting is as follows: 
if a family of density operators
$(\rth)_{\thin}$ commutes with 
a complete set of projections $(P_j) ,$
the corresponding projective measurement does not disturb the states
$(\rth)_{\thin} .$

This motivates us to ask the following question:
how much information can we extract without disturbing 
a family of partially known quantum states?
The present paper addresses this problem in the von Neumann algebra framework.
We show that if we do not disturb a statistical experiment, 
we can access at most what we call the classical part of the statistical experiment,
even if we allow non-classical outcome spaces.
We mention that this problem was also considered 
in \cite{PhysRevA.66.022318}
for finite-dimensional density operators.

This paper is organized as follows.
After mathematical preliminaries in Section~\ref{sec:prel},
we introduce in Section~\ref{sec:accessible} the classical part
of a statistical experiment,
which is defined as the restriction of 
the equivalent minimal sufficient statistical experiment
to the center of the outcome space (Definition~\ref{defi:clpart}).
We show that we can access at most the classical part
by a Schwarz channel
without disturbing a given statistical experiment (Theorem~\ref{theo:main}).
From this, the no-broadcasting theorem immediately follows 
(Section~\ref{subsec:nb}, Corollary~\ref{coro:nb}).
In Section~\ref{subsec:KI}, 
we consider the case of density operators
on a (possibly infinite-dimensional, or even non-separable) Hilbert space
and find that the classical part in this case corresponds to 
the classical probability distributions appearing in the maximal decomposition
in \cite{PhysRevA.66.022318}.
In Section~\ref{sec:direct}, we show that 
the classical part of the direct product of two statistical experiments
coincides with the direct product of the classical parts of the statistical experiments
(Theorem~\ref{theo:directcl}).
The proof of Theorem~\ref{theo:directcl} is 
based on Theorem~\ref{theo:directms}
which states that 
the direct product of two minimal sufficient statistical experiments 
is also minimal sufficient.
Finally, Section~\ref{sec:conclusion} concludes the paper.

\section{Preliminaries}\label{sec:prel}
In this section, we introduce mathematical preliminaries 
on operator algebras, channels between them, 
and (minimal sufficient) statistical experiment.
For a general reference of operator algebras, we refer
\cite{takesakivol1}.

\subsection{States and channels on operator algebras.}
\label{subsec:ch}
Throughout this paper, we only consider 
$\ast$-, \cstar-, and von Neumann algebras
with unit elements in the multiplications.
The unit element of a $\ast$-algebra $\A$
is denoted by 
$\unit_\A .$
We denote by $\id_\A$ the identity map on a $\ast$-algebra $\A $
and by $\Mn (\A)$ the set of $n \times n$ matrices 
with entries from $\A .$
We denote by 
$\zent (\A)$
the center
$\set{ Z \in \A | AZ = ZA   \, (\forall A \in \A)} $
of $\A .$
The algebra of bounded operators on a Hilbert space 
$\cH$ is denoted by $\LH .$

A positive linear functional on a \cstar-algebra $\A$ 
satisfying the normalization condition 
$\vph (\unit_\A) =1$
is called a state on $\A .$
A state $\vph$ on a von Neumann algebra $\M$
is called normal if $\vph (\sup_i A_i ) = \sup_i \vph (A_i)$ 
for any monotonically increasing bounded net 
$(A_i)$ on $\M .$
A state on $\M$ is normal if and only if $\vph$ is continuous 
in the ultraweak topology 
on $\M .$
The sets of ultraweakly continuous linear functionals and
normal states on $\M$ are denoted by
$\M_\ast$
and
$\Ss (\M) ,$ respectively.
The support 
$\s (\vph)$
of 
$\vph \in \Ss (\M)$
is the minimal projection 
$\s (\vph) \in \M$
satisfying
$\vph (\s (\vph))=1.$
For each $\vph \in \Ss ( \M )$ and each $A \in \M ,$
we have
\[
	\vph (A)
	=
	\vph (\s (\vph) A \s (\vph) ) .
\]
For positive $A \in \M ,$
$\vph (A) = 0$ implies $\s(\vph)A \s(\vph) = 0.$

Let $\A$ be a \cstar-algebra and let $ \pi_\A$ 
be the universal representation of $\A $
acting on the Hilbert space $\cH_\A ,$
i.e.\
$\pi_\A \colon \A \to \calL (\cH_\A)$ 
is the direct sum of the GNS representations
taken over all the states on $\A .$
The von Neumann algebra $\pi_\A (\A)^{\prime\prime} ,$
where the prime denotes the commutant,
is called the enveloping von Neumann algebra of $\A$
and denoted by $\A^{\ast\ast} .$
The enveloping von Neumann algebra $\A^{\ast\ast}$
is, as a Banach space, 
isometrically isomorphic to the double dual of $\A .$
Every linear functional $\phi \in \A^\ast$
extends to a linear functional 
$\overline{\phi} \in (\A^{\ast\ast})_\ast ,$
i.e.\
$\overline{\phi} $ is the ultraweakly continuous linear functional 
on $\A^{\ast\ast}$ satisfying
$\phi = \overline{\phi} \circ \pi_\A .$
The enveloping von Neumann algebra satisfies the following 
universal property:
for each representation 
$\pi \colon \A \to \LK$ 
acting on the Hilbert space $\cK ,$
there exists a normal representation 
$\overline{\pi} \colon \A^{\ast\ast} \to \LK$
that is an extension of $\pi,$
i.e.\
$\pi = \overline{\pi} \circ \pi_\A .$
If there is no confusion, we identify $\A$
with the \cstar-subalgebra $\pi_\A (\A) $ 
of $\A^{\ast\ast}  .$

Next, we introduce channels.
In this paper, we consider channels whose outcome spaces are 
algebraic tensor products of von Neumann algebras and 
do not in general have unique \cstar-norms.
For this reason, we slightly generalize the notion of channels as follows.
Let $\A$ be a $\ast$-algebra,
let $\B$ be a \cstar-algebra,
and
let $\Lambda \colon \A \to \B$ a linear map.
$\Lambda$ is called unital if
$\Lambda (\unit_\A) = \unit_\B .$
For an integer $n \geq 1 ,$
$\Lambda$ is called $n$-positive if
\[
	\sum_{i,j =1}^n
	B_i^\ast
	\Lambda (A_i^\ast A_j)
	B_j
	\geq 0
\]
for each $A_i \in \A$ and each $B_i \in \B $
$(i=1,\dots , n) .$
$\Lambda$ is called completely positive (CP)
if $\Lambda$ is $n$-positive for every positive integer $n .$
A unital map $\Lambda \colon \A \to \B $
is called Schwarz if 
$\Lambda$ satisfies
\begin{gather*}
	\Lambda (A^\ast) = \Lambda (A)^\ast ,
	\\
	\Lambda (A^\ast A) \geq \Lambda (A^\ast) \Lambda (A)
\end{gather*}
for all $A \in \A .$
The Schwarz condition is equivalent to the following single matrix inequality:
\[
	\begin{pmatrix}
	\Lambda  (A^\ast A) & \Lambda (A^\ast) \\
	\Lambda (A) & \unit_\B
	\end{pmatrix}
	\geq 0.
\]
A unital Schwarz map is positive, i.e.\ $1$-positive.
A unital positive map 
$\Lambda \colon \A \to \B$
is called a channel (in the Heisenberg picture).
The codomain $\B$ and domain $\A$ of a channel
$\Lambda \colon \A \to \B$
are called the input and outcome spaces, or algebras, 
of $\Lambda ,$
respectively.
A channel 
$\Lambda \colon \A \to \B$
maps a state $\phi $ on the input space $\B$ 
to the normalized positive linear 
functional $\phi \circ \Lambda$ on $\A .$
The sets of CP and Schwarz channels from $\A$ to $\B$
are denoted by 
$\cpchset{\A}{\B}$
and
$\schset{\A}{\B} ,$
respectively.
As in the case of channels with outcome 
\cstar-algebras~\cite{choi1974},
we can show that
any $2$-positive (and therefore any CP) channel $\Lambda \colon \A \to \B$
is a Schwarz channel.
If either $\A$ or $\B$
is a commutative \cstar-algebra,
a linear map $\Lambda \colon \A \to \B $
is CP if and only if $\Lambda$ is positive.

Let $\Lambda \in \schset{\A}{\B}$ be a Schwarz channel
with input \cstar-algebra $\B$ and outcome $\ast$-algebra $\A .$
We define the \emph{multiplicative domain} of $\Lambda$ by
\[
	\M_\Lambda 
	:=
	\set{A \in \A | 
	\Lambda (A^\ast A )= \Lambda (A^\ast) \Lambda (A), \,
	\Lambda (A A^\ast )= \Lambda (A) \Lambda (A^\ast)
	} .
\]
The following lemma for the multiplicative domain of a Schwarz channel
can be shown analogously as in ~\cite{doi:10.1142/S0129055X11004412}
(Lemma~3.9).
\begin{lemma}
\label{lemm:mdomain}
Let $\A$ be a $\ast$-algebra, let $\B$ be a \cstar-algebra,
and let $\Lambda \in \schset{\A}{\B}$ be a Schwarz channel.
Then for $A \in \A ,$
$A \in \M_{\Lambda} $ if and only if
$\Lambda (BA) = \Lambda (B) \Lambda (A)$
and 
$\Lambda (AB) = \Lambda (A) \Lambda (B)$
for all $B \in \A .$
\end{lemma}
From Lemma~\ref{lemm:mdomain}, the multiplicative domain
$\M_\Lambda$ 
is a unital $\ast$-subalgebra of $\A .$

Let $\M$ and $\N$ be von Neumann algebras.
A channel $\Lambda \colon \M \to \N $ is called normal 
if it is continuous in the ultraweak topologies of $\M$ and $\N ,$
respectively.
The set of normal CP (respectively, Schwarz) channels 
from $\M$ to $\N$ is denoted by 
$\ncpchset{\M}{\N}$
(respectively, $\nschset{\M}{\N}$).
We also write 
$\ncpch{\M} := \ncpchset{\M}{\M}$
and
$\nsch{\M} := \nschset{\M}{\M} .$

\subsection{Tensor products of operator algebras}
\label{subsec:tensor}
Let $\A$ and $\B$ be \cstar-algebras.
We denote by $\A \atensor \B$ the algebraic tensor product of 
$\A$ and $\B .$
The structure of a 
$\ast$-algebra is naturally induced to
the algebraic tensor product $\A \atensor \B$
by
\begin{gather*}
	\left(
	\sum_i A_i \otimes B_i
	\right)
	\left(
	\sum_j C_j \otimes D_j
	\right)
	:=
	\sum_{i,j}
	A_i C_j \otimes B_i D_j ,
	\\
	\left(
	\sum_i A_i \otimes B_i
	\right)^\ast
	:=
	\sum_i A_i^\ast \otimes B_i^\ast
\end{gather*}
$(A_i , C_j \in \A; \,B_i , D_j \in \B ).$
A norm $\gnorm{\cdot}$ on $\A \atensor \B$ satisfying
$\gnorm{XY} \leq \gnorm{X} \gnorm{Y} , $
$\gnorm{X^\ast} = \gnorm{X} ,$
and
$\gnorm{X^\ast X} = \gnorm{X}^2 $
$(X ,Y \in \A \atensor \B)$
is called a \cstar-norm on $\A \atensor \B .$
The completion of $\A \atensor \B$ with respect to
the \cstar-norm $\gnorm{\cdot} ,$ 
written as $\A \gtensor \B ,$
is a \cstar-algebra and called a \cstar-tensor product of 
$\A$ and $\B .$
We denote the injective norm, the injective tensor product,
the projective norm, and the projective tensor product 
by
$\minnorm{\cdot} ,$
$\A \mintensor \B,$
$\maxnorm{\cdot} ,$
and 
$\A \maxtensor \B ,$
respectively.
We have
$\minnorm{X} \leq \gnorm{X} \leq \maxnorm{X}$
$(X\in \A \atensor \B)$
for any \cstar-norm $\gnorm{\cdot}$
on $\A \atensor \B .$
If either $\A$ or $\B$ is commutative,
then we have $\minnorm{\cdot} = \maxnorm{\cdot}$
and the \cstar-norm is unique in this case.

Let $\M$ and $\N$ be von Neumann algebras acting on
Hilbert spaces $\cH$ and $\cK ,$ respectively.
The von Neumann algebra 
on $\cH \otimes \cK$
generated by 
$\{ A \otimes B | A \in \M , \, B \in \N \}$
is called the normal tensor product of $\M $ and $\N ,$
denoted by $\M \ntensor \N .$

Let $\M$ and $\N$ be von Neumann algebras,
let $\B$ be a \cstar-algebra,
and let $\M \otimes_y \B$ be the algebraic or a \cstar-tensor product
of $\M$ and $\B .$
A linear map
$\Lambda \colon \M \otimes_y \B \to \N$
is called left-normal 
if the linear map
\[
	\Lambda_B
	\colon 
	\M \ni A \mapsto
	\Lambda (A \otimes B)
	\in \N
\]
is ultraweakly continuous for each $B \in \B .$

Let $\M_1 , \M_2 ,$ and $\N$ be von Neumann algebras
and let $\M_1 \otimes_y \M_2$ be the algebraic or a \cstar-tensor product
of $\M_1$ and $\M_2 .$
A linear map $\M_1 \otimes_y \M_2 \colon \to \N$
is called binormal if the linear maps
\begin{gather*}
	\Lambda_{A_2^\prime}^L
	\colon
	\M_1 \ni A_1
	\mapsto
	\Lambda (A_1 \otimes A_2^\prime)
	\in \N ,
	\\
	\Lambda_{A_1^\prime}^R
	\colon
	\M_2 \ni A_2
	\mapsto
	\Lambda (A_1^\prime \otimes A_2)
	\in \N 
\end{gather*}
are ultraweakly continuous for each 
$A_1^\prime \in \M_1$
and each
$A_2^\prime \in \M_2 .$

\begin{proposition}[\cite{takesakivol1}, Propositions IV.4.23 and IV.5.13]
\label{prop:cp}
Let $\A , \B , \C , $ and $\D$ be \cstar-algebras
and let $\Lambda \in \cpchset{\A}{\B}$ and 
$\Gamma \in \cpchset{\C}{\D}$ be CP channels.
Define the algebraic tensor product map
$\Lambda \atensor \Gamma \colon
\A \atensor \C \to \B \atensor \D
$
by
\[
	\Lambda \atensor \Gamma 
	\left(
	\sum_i A_i \otimes C_i
	\right)
	:=
	\sum_i
	\Lambda (A_i )
	\otimes 
	\Gamma (B_i)
\]
$(A_i \in \A , C_i \in \C) .$
Then $\Lambda \atensor \Gamma$ uniquely extends to a 
CP channel 
$\Lambda \mintensor \Gamma \in \cpchset{\A \mintensor \C}{\B \mintensor \D} .$
If we further assume that $\A, \B , \C ,$ and $\D$ are von Neumann algebras
and that $\Lambda$ and $\Gamma$ are normal,
then $\Lambda \atensor \Gamma $ uniquely extends to a normal channel
$\Lambda \ntensor \Gamma \in \ncpchset{\A \ntensor \C}{\B \ntensor \D} .$
\end{proposition}

A special case of Proposition~\ref{prop:cp} is normal states on von Neumann algebras:
if $\vph$ and $\psi$ are normal states on von Neumann algebras
$\M$ and $\N ,$ respectively,
then 
the linear functional $\vph \atensor \psi$ on $\M \atensor \N$
uniquely extends to a normal state $\vph \ntensor \psi $
on $\M \ntensor \N .$

\subsection{Coarse-graining relations for statistical experiments}
\label{subsec:cg}
A statistical experiment is a triple
$\E = \seE $
such that
$\M$ is a von Neumann algebra,
$\Theta$ is a set,
and 
$(\phth)_{\thin} \in \Ss (\M)^\Theta$
is a family of normal states on $\M$
indexed by $\Theta .$
$\M$ and $\Theta$ are called the outcome space and the parameter set 
of $\E ,$ respectively.
Operationally, a statistical experiment 
$\E = \seE$
corresponds to 
the situation in which the state on the outcome space 
$\M$ is known to 
be one of the states $(\phth)_{\thin} .$
The family of normal states
$(\phth)_{\thin}$ is called faithful if
for any positive $A \in \M ,$
$\phth (A) = 0$ for all $\thin$ implies
$A = 0.$
If so, we also say that $\E$ is faithful.
$\E$ is faithful if and only if
\[
	\bigvee_{\thin} \s (\phth) = \unit_\M .
\]
A statistical experiment $\E$ is called classical if 
the outcome space of $\E$ is commutative.

For a statistical experiment
$\E = \seE$
we define the normal extension of $\E$
by 
$
\overline{\E}
:=
(\M^{\ast\ast} , \Theta , (\overline{\phth})_{\thin}) ,
$
where $\overline{\phth}$ is the normal extension of 
$\phth$
to 
$\M^{\ast\ast} .$

Let $\E = \seE$
and $\F = \seF$
be statistical experiments 
with a common parameter set $\Theta .$
We introduce the following coarse-graining
(or randomization) and isomorphism relations for statistical experiments:
\begin{itemize}
\item[$\bullet$]
$\E \cocp \F $
($\E$ is a coarse-graining of $\F$)
$	:\defarrow
	\exists \Phi \in \cpchset{\M}{\N}
$
s.t.\
$
	[
	\phth = \psth \circ \Phi
	(\forall \thin)
	] ;
$
\item[$\bullet$]
$\E \concp \F $
($\E$ is a normal coarse-graining of $\F$)
$	:\defarrow
	\exists \Phi \in \ncpchset{\M}{\N}
$
s.t.\
$
	[
	\phth = \psth \circ \Phi
	(\forall \thin)
	] ;
$
\item[$\bullet$]
$\E \eqcp \F $
($\E$ is CP equivalent to $\F$)
$:\defarrow$
$\E \cocp \F$ and $\F \cocp \E ;$
\item[$\bullet$]
$\E \eqncp \F $
($\E$ is normally CP equivalent to $\F$)
$:\defarrow$
$\E \concp \F$ and $\F \concp \E  ;$
\item[$\bullet$]
$\E \cong \F$ 
($\E$ and $\F$ are normally isomorphic)
$:\defarrow$
there exists a normal isomorphism
$\pi \colon \M \to \N$
such that
$\phth = \psth \circ \pi $
for all $\thin .$
\end{itemize}
The relations $\cocp$ and $\concp$ are binary preorder relations for statistical experiments,
and $\eqcp ,$ $\eqncp ,$ and $\cong$ are binary equivalence relations.

The following two lemmas, which will be used in the proof of
Theorem~\ref{theo:main},
are due to \cite{kuramochi2017incomp}
(Corollaries~3 and 4).

\begin{lemma}
\label{lemm:nex}
Let $\E = \seE$ be a statistical experiment
and let $\overline{\E} = (\M^{\ast\ast }, \Theta , (\overline{\phth})_{\thin})$ 
be the normal extension of $\E .$
Then $\E \eqncp \overline{\E} .$
\end{lemma}

\begin{lemma}
\label{lemm:coeq}
Let $\E = \seE$ and $\F = \seF$ be statistical experiments
with a common parameter set $\Theta .$
Then $\E \cocp \F$ 
(respectively, $\E \eqcp \F$)
if and only if
$\E \concp \F$ 
(respectively, $\E \eqncp \F$).
\end{lemma}
The proof in \cite{kuramochi2017incomp} is based on
the results on the compatibility relations of CP channels
and requires a number of irrelevant discussions.
For readers' convenience, we give other
direct proofs of these lemmas in Appendix~\ref{sec:appendix}.

\subsection{Minimal sufficient subalgebra and statistical experiment}
Let $\A$ be a \cstar-algebra and let 
$\B$ be a \cstar-subalgebra of $\A .$
A linear map
$\condi \colon \A \to \B$
is called a conditional expectation onto $\B$ if
$\condi (B) = B$
for all $B \in \B $
and $\norm{\condi (A)} \leq \norm{A}$
for all $A \in \A .$
Tomiyama's theorem (e.g.\ \cite{brown2008c}, Theorem~1.5.10)
states that 
a conditional expectation $\condi$ onto $\B$ satisfies
$
\condi (B_1 A B_2)
=
B_1
\condi (A) 
B_2 
$
$(A \in \A ; B_1 , B_2 \in \B).$
From this we also have
$ \condi \in \cpchset{\A}{\B} .$

Let $\E = \seE$ be a statistical experiment and let $\N \subseteq \M$
be a von Neumann subalgebra of $\M .$
$\N$ is said to be sufficient 
(in the sense of Umegaki~\cite{umegaki1959,umegaki1962,luczak2014})
with respect to the family
$(\phth)_{\thin}$
if there exists a normal conditional expectation
$\condi$
from $\M$ onto $\N$
such that
$\phth \circ \condi = \phth$
for all $\thin .$
$\N$ is said to be a minimal sufficient subalgebra 
with respect to $(\phth)_{\thin}$
if 
$\N$ is sufficient and contained in any sufficient subalgebra.
A minimal sufficient subalgebra is, if exists, unique
for a given family of normal states.

A statistical experiment $\E = \seE$ is called 
\emph{minimal sufficient}
if $\E$ satisfies either of the following equivalent conditions 
(\cite{kuramochi2017minimal}, Theorem~2):
\begin{enumerate}[(i)]
\item
for $\Phi \in \ncpch{\M} ,$
$\phth \circ \Phi = \phth$
for all $\thin$ implies
$\Phi = \id_\M ;$
\item
for $\Phi \in \nsch{\M} ,$
$\phth \circ \Phi = \phth$
for all $\thin$ implies
$\Phi = \id_\M ;$
\item
$(\phth)_{\thin}$ is faithful and
$\M$ is a minimal sufficient subalgebra with respect to
$(\phth)_{\thin} .$
\end{enumerate}
From condition~(iii), 
a minimal sufficient statistical experiment is faithful.

Let $\E = \seE$ be an arbitrary statistical experiment.
It is known
that there exists a minimal sufficient statistical experiment 
$\E_0$ satisfying $\E \eqncp \E_0 $
(\cite{kuramochi2017minimal}, Theorem~1).
Such a minimal sufficient statistical experiment $\E_0$ is unique 
up to normal isomorphism 
and, in this sense, we may say that
$\E_0$ is \emph{the} minimal sufficient statistical experiment
normally CP equivalent to $\E .$
If $\E$ is faithful, $\E_0$ can be constructed as follows~\cite{kuramochi2017minimal,luczak2014}.
Define 
\begin{gather*}
	\mathbb{F}
	:=
	\set{
	\Phi \in \ncpch{\M} 
	|
	\phth \circ \Phi
	=\phth 
	\,
	(\forall \thin)
	},
	\\
	\M_0
	:=
	\set{
	A \in \M |
	\Phi (A) 
	= A
	\,
	(\forall \Phi \in \mathbb{F})
	}.
\end{gather*}
Then from the mean ergodic theorem for von Neumann algebras~\cite{ISI:A1979HD47500016},
there exists a normal conditional expectation $\condi$
from $\M$ onto $\M_0$ such that
\begin{gather}
	\phth \circ \condi = \phth
	\quad
	(\forall \thin) ,
	\label{eq:unchange}
	\\
	\condi \circ \Phi = \Phi \circ \condi = \condi
	\quad
	(\forall \Phi \in \mathbb{F}) .
	\label{eq:ergodic}
\end{gather}
If we write as $\phthz$ the restriction of $\phth$ to $\M_0 ,$
then we can show that $\E_0 := \seEz $
is the desired minimal sufficient statistical experiment.
In this case, $\M_0$ is the minimal sufficient subalgebra with respect to
$(\phth)_{\thin} .$
We note that the condition~\eqref{eq:unchange}
for a conditional expectation $\condi$ onto $\M_0$
uniquely determines $\condi ;$
indeed if $\condi^\prime$ is a normal conditional expectation
from $\M$ onto $\M_0$ satisfying
$\phth \circ \condi^\prime = \phth$
for all $\thin ,$
then condition~\eqref{eq:ergodic}
implies 
$\condi^\prime = \condi \circ \condi^\prime = \condi .$

\section{Accessible part of statistical experiment}
\label{sec:accessible}
In this section, we prove the main result (Theorem~\ref{theo:main})
and apply it to the no-broadcasting theorem and 
the case of density operators on a
Hilbert space.

\subsection{Accessible and classical parts of a statistical experiment}
\label{subsec:main}

We first define the accessibility relations for statistical experiments as follows.

\begin{definition}[accessibility relations for statistical experiments]
\label{defi:acrelation}
Let $\E = (\M , \Theta , (\phth)_{\thin})$ 
and
$\F = (\N , \Theta , (\psth)_{\thin})$
be statistical experiments with a common parameter set
$\Theta ,$
let $X$ be either $\CP$ or $\Sch ,$ and
let $\M \otimes_y \N$ be the algebraic or a \cstar-tensor product of 
$\M$ and $\N .$
\begin{enumerate}
\item
A channel $\Lambda \in \mathbf{Ch}^X( \M \otimes_y \N \to \M   )$
is said to 
\emph{extract $\F$ without disturbing $\E$}
if for each $\thin ,$ each $A \in \M ,$ and each $B \in \N , $
it holds that
\begin{gather*}
	\phth (A)
	=
	\phth \circ \Lambda (A \otimes \unit_\N) ,
	\\
	\psth (B)
	=
	\phth \circ \Lambda (\unit_\M \otimes B ) .
\end{gather*}
\item
$\F$ is said to be \emph{accessible without disturbing $\E$}
in the sense of $X$ channel and the tensor product $\otimes_y ,$
written as $\F \clp^X_{y} \E ,$
if there exists a channel 
$\Lambda \in \mathbf{Ch}^X( \M \otimes_y \N \to \M   )$
that extracts $\F$ without disturbing $\E .$
\item
We write $\F \clp^X_{y, \bin} \E$
(respectively, $\F \clp^X_{y , \lnor}\E$)
if there exists a binormal (respectively, left-normal)
channel
$\Lambda \in \mathbf{Ch}^X( \M \otimes_y \N \to \M   )$
that extracts $\F$ without disturbing $\E .$
\end{enumerate}
\end{definition}

Let us consider a special case 
in which $\M $ and $\N  $ are the full operator algebras
$\LH$ and $\LK$
on finite-dimensional Hilbert spaces $\cH$ and $\cK ,$ respectively.
As usual, we identify the normal states
$\phth$ and $\psth$ 
with density operators $\rth$ on $\cH$
and $\sigma_\theta$ on $\cK ,$ respectively, such that
$\phth (A) = \tr (\rth A)$
and 
$\psth (B) = \tr (\sigma_\theta B)$
$(A \in \LH , B \in \LK) .$
Then a CP channel 
$\Lambda \in \cpchset{\calL (\cH \otimes \cK)}{\LH}$
extracts $\F$ without disturbing $\E$ if and only if
\begin{equation}
	\rth = \tr_\cK \circ \Lambda^\ast (\rth) ,
	\quad
	\sigma_\theta = \tr_\cH \circ \Lambda^\ast (\rth) 
	\label{eq:accessible}
\end{equation}
for all $\thin ,$
where $\tr_\cH$ and $\tr_\cK$ denote the partial traces over 
$\cH$ and $\cK ,$ respectively.
If $\F$ is identical to $\E ,$
the condition \eqref{eq:accessible}
reduces to the broadcastability condition for 
the states $(\rth)_{\thin}$
and the channel $\Lambda .$

Next, we introduce the classical part of a statistical experiment.

\begin{definition}[classical part of statistical experiment]
\label{defi:clpart}
Let $\E = \seE$ be a statistical experiment
and let 
$\E_0 = \seEz$ 
be the minimal sufficient statistical experiment
satisfying $\E \eqncp \E_0 .$
We define \emph{the classical part} of $\E$
as the statistical experiment
$
(\zent (\M_0) , \Theta , (\phth^{\cl})_{\thin}) ,
$
where $\phth^{\cl}$ is the restriction of $\phth^{(0)}$
to the center $\zent (\M_0)  .$
The classical part of $\E$ is denoted by
$\E_\cl .$
\end{definition}
We note that the classical part $\E_\cl$ of a statistical experiment
$\E$ is well-defined up to normal isomorphism
due to the uniqueness of the minimal sufficient
$\E_0 .$

The following theorem,
the main result of this paper,
states that 
we can access at most the classical part of a statistical experiment $\E$
without disturbing $\E ;$
in other words, the classical part $\E_\cl$ is the accessible part of $\E .$
\begin{theorem}
\label{theo:main}
Let 
$\E = \seE,$
$\F = \seF ,$
$X ,$
and
$\M \otimes_y \N $
be the same as in Definition~\ref{defi:acrelation},
and let $z$ be either $\bin$ or $\lnor $ or blank.
Then the condition $\F \clp^X_{y, z} \E$
does not depend on the choices of $X ,y,$ and $z ,$
and is equivalent to $\F \concp \E_\cl .$
%
%
%
%
%
%
%
%
\end{theorem}

To prove Theorem~\ref{theo:main}, 
we first show the following four lemmas.

\begin{lemma}
\label{lemm:lnor}
Let $\A , \B ,$ and $\C$ be \cstar-algebras 
and let $\Lambda \in \schset{\A \atensor \B}{\C}$
be a Schwarz channel.
Then there exists a left-normal Schwarz channel 
$\widetilde{\Lambda} \in \schset{\A^{\ast \ast} \atensor \B}{\C^{\ast \ast}}$ 
that is an extension of $\Lambda .$
\end{lemma}
\begin{proof}
For each $A \in \A$ and each $B \in \B ,$ we have
\begin{align*}
	\norm{\Lambda (A \otimes B)}
	&=
	\norm{
	\Lambda ( A \otimes B )^\ast
	\Lambda ( A \otimes B )
	}^{1/2}
	\\
	& \leq 
	\norm{ 
	\Lambda ( ( A \otimes B )^\ast  ( A \otimes B )  )
	}^{1/2}
	\\
	&=
	\norm{\Lambda ( A^\ast A \otimes B^\ast B )}^{1/2}
	\\
	& \leq
	\norm{\Lambda ( \norm{A}^2 \unit_\A \otimes B^\ast B )}^{1/2}
	\\
	& \leq
	\norm{\Lambda ( \norm{A}^2 \unit_\A \otimes \norm{B}^2 \unit_\B )}^{1/2}
	\\
	&=
	\norm{A} \norm{B} .
\end{align*}
Therefore, for each $B \in \B ,$ the linear map
\[
	\Lambda_B \colon 
	\A \ni A \mapsto \Lambda (A \otimes B) \in \C
\]
is bounded and hence its double dual map
$
	\Lambda_B^{\ast \ast}
	\colon \A^{\ast \ast} \to \C^{\ast\ast}
$
is an ultraweakly continuous linear map 
that is an extension of $\Lambda_B .$
Since $\B \ni B \mapsto \Lambda_B$ is linear with respect to $B  ,$
also is $ \B \ni B \mapsto  \Lambda^{\ast \ast}_B .$
Thus we can define a linear map 
$\widetilde{\Lambda} \colon \A^{\ast \ast} \atensor\B \to \C^{\ast\ast}$ 
by
$\widetilde{\Lambda} (
\sum_i
A_i^{\prime \prime} \otimes B_i
)
:=
\sum_i
\Lambda_{B_i}^{\ast\ast }(A^{\prime\prime}_i)
$
$(A_i^{\prime\prime} \in \A^{\ast\ast} , B_i \in \B) .$
Then $\widetilde{\Lambda}$ is left-normal and an extension of $\Lambda .$
Thus it is sufficient to show the Schwarz condition
\begin{equation}
	\begin{pmatrix}
	\widetilde{\Lambda} (X^\ast X) & \widetilde{\Lambda} (X^\ast) \\
	\widetilde{\Lambda} (X) & \unit_{\C}
	\end{pmatrix}
	\geq 0
	\label{eq:scht}
\end{equation}
for each $X \in \A^{\ast \ast}\atensor \B .$
We write $X$ as 
$X = \sum_{i=1}^n A^{\prime\prime}_i \otimes B_i$
$(A^{\prime\prime}_i \in \A^{\ast\ast} , B_i \in \B) .$
Then, from the Kaplansky density theorem, there exist nets 
$(A_{i\alpha})$ $(i=1,\dots , n)$ on $\A$ satisfying
$\norm{A_{i\alpha}} \leq \norm{A_i^{\prime\prime}}$
and
$A_{i\alpha} \xrightarrow{us\ast} A_i^{\prime\prime} ,$
where $\xrightarrow{us\ast}$ denotes the convergence in 
the ultrastrong$\ast$ topology on $\cH_{\A} .$
If we put $X_\alpha := \sum_i A_{i\alpha} \otimes B_i ,$ 
then
\begin{gather*}
	\Lambda (X_\alpha)
	=
	\sum_{i=1}^n \widetilde{\Lambda} (A_{i\alpha} \otimes B_i)
	\xrightarrow{uw}
	\sum_{i=1}^n \widetilde{\Lambda} (A_{i}^{\prime\prime} \otimes B_i)
	=
	\widetilde{\Lambda}(X) ,
	\\
	\Lambda(X^\ast_\alpha)
	\xrightarrow{uw}
	\widetilde{\Lambda}(X^\ast) ,
	\\
	A^\ast_{i\alpha} A_{j\alpha}
	\xrightarrow{us\ast}
	A^{\prime\prime \ast}_i
	A^{\prime\prime}_j ,
	\\
	\Lambda (X^\ast_\alpha X_\alpha)
	=
	\sum_{i,j=1}^n
	\widetilde{\Lambda}
	(A_{i\alpha}^\ast A_{j\alpha} \otimes B^{\ast}_i B_j)
	\xrightarrow{uw}
	\sum_{i,j=1}^n
	\widetilde{\Lambda}
	(A_{i}^{\prime\prime \ast} A_{j}^{\prime\prime} \otimes B^{\ast}_i B_j)
	=
	\widetilde{\Lambda} (X^\ast X) ,
\end{gather*}
where $\xrightarrow{uw}$ denotes the convergence in the ultraweak topology
on $\cH_\C  .$
Since $\Lambda$ is Schwarz, we obtain
\[
	0
	\leq
	\begin{pmatrix}
	\Lambda(X^\ast_\alpha X_\alpha) & \Lambda (X^\ast_\alpha) \\
	\Lambda (X_\alpha) & \unit_\C
	\end{pmatrix}
	\xrightarrow{uw}
	\begin{pmatrix}
	\widetilde{\Lambda}(X^\ast X) & \widetilde{\Lambda} (X^\ast) \\
	\widetilde{\Lambda} (X) & \unit_\C
	\end{pmatrix}
\]
on $\cH_\C \otimes \cmplx^2 ,$
which implies the Schwarz condition~\eqref{eq:scht}.
\qed
\end{proof}

\begin{lemma}
\label{lemm:schwarz}
Let $\A_1, \A_2 , \B_1 , \B_2 ,$ and $\C$ be \cstar-algebras,
let $\Phi \in \cpchset{\A_1}{\A_2}$
and 
$\Psi \in \cpchset{\B_1}{\B_2}$
be CP channels,
and let $\Lambda \in \schset{\A_2 \atensor \B_2 }{ \C}$
be a Schwarz channel.
Then 
$\Lambda \circ (\Phi \atensor \Psi) \in 
\schset{\A_1 \atensor \B_1}{\C} .$
\end{lemma}
\begin{proof}
We first show that 
$\widetilde{\Lambda} := \Lambda \circ (\Phi \atensor \id_{\B_2})
\colon \A_1 \atensor \B_2 \to \C
$
is Schwarz.
For this we have to show the Schwarz inequality 
\begin{equation}
	\Lambda \circ \widetilde{\Phi} (X^\ast X)
	\geq 
	\Lambda \circ \widetilde{\Phi} (X^\ast )
	\Lambda \circ \widetilde{\Phi} (X )
	\notag
\end{equation}
for each $X \in \A_1 \atensor \B_2 ,$
where $\widetilde{\Phi} := \Phi \atensor \id_{\B_2} .$
From the Schwarz condition of $\Lambda ,$
we have
\[
	\Lambda \left(
	\widetilde{\Phi} (X^\ast ) 
	\widetilde{\Phi} (X)
	\right)
	\geq
	\Lambda \circ
	\widetilde{\Phi} (X^\ast )
	\Lambda \circ
	\widetilde{\Phi} (X) .
\]
Thus it is sufficient to show 
\begin{equation}
	\Lambda \left(
		\widetilde{\Phi} (X^\ast  X)
		-
		\widetilde{\Phi} (X^\ast )
		\widetilde{\Phi} (X)
	\right)
	\geq
	0.
	\label{eq:schineq}
\end{equation}
We write $X$ as 
$X = \sum_{i=1}^n A_i \otimes B_i$
$(A_i \in \A_1 , B_i \in \B_2) .$
Since $\Phi$ is a CP channel, the map
\[
	\Phi_n
	\colon
	\Mn (\A_1)
	\ni
	(A_{ij})_{i,j =1}^n
	\mapsto
	(\Phi (A_{ij} ) )_{i,j =1}^n
	\in 
	\Mn (\A_2)
\]
is also a CP channel.
Thus from the Schwarz inequality of $\Phi_n ,$
the matrix
$
(
\Phi (A_i^\ast A_j) 
-
\Phi (A_i^\ast) \Phi ( A_j)
)_{i,j =1}^n
\in
\Mn (\A_2)
$
is positive.
Therefore there exists a matrix $(C_{ij})_{i,j =1}^n \in \Mn (\A_2)$
such that
\[
	\Phi (A_i^\ast A_j) 
	-
	\Phi (A_i^\ast) \Phi ( A_j)
	=
	\sum_{k=1}^n
	C_{ki}^\ast 
	C_{kj} .
\]
Thus we have
\begin{align*}
	\widetilde{\Phi} (X^\ast X)
	-
	\widetilde{\Phi} (X^\ast)
	\widetilde{\Phi} (X)
	&=
	\sum_{i,j =1}^n
	(
	\Phi (A_i^\ast A_j) 
	-
	\Phi (A_i^\ast) \Phi ( A_j)
	)
	\otimes 
	B_i^\ast B_j
	\\
	&=
	\sum_{i,j,k =1}^n
	C_{ki}^\ast C_{kj}
	\otimes 
	B_i^\ast B_j
	\\
	&=
	\sum_{k=1}^n
	Y_k^\ast Y_k ,
\end{align*}
where $Y_k := \sum_{i=1}^n C_{ki} \otimes B_i .$
Consequently,
\[
	(\text{LHS of \eqref{eq:schineq}})
	=
	\sum_{k=1}^n 
	\Lambda (Y_k^\ast Y_k)
	\geq 
	\sum_{k=1}^n 
	\Lambda (Y_k^\ast )
	\Lambda (Y_k)
	\geq 0.
\]
Therefore $\widetilde{\Lambda}$ is Schwarz.

We can analogously show that 
$
\Lambda  \circ (\Phi \atensor \Psi)
=
\widetilde{\Lambda}
\circ
(\id_{\A_1} \atensor \Psi)
$
is also Schwarz,
which completes the proof.
\qed
\end{proof}

\begin{lemma}
\label{lemm:coarse}
Let 
$\E = (\M , \Theta , (\phth)_{\thin}),$
$\E_1 = (\M_1 , \Theta , (\phtho)_{\thin}) ,$
$\F = (\N , \Theta , (\psth)_{\thin}),$
and 
$\F_1 = (\N_1 , \Theta , (\pstho)_{\thin})$
be statistical experiments.
Then we have the following.
\begin{enumerate}
\item
$\F \clp^{\CP}_{\min , \bin} \E$ 
and 
$\F_1 \concp \F$
imply 
$\F_1 \clp^{\CP}_{\min , \bin} \E  .$
\item
$\F \clp^{\CP}_{\min , \bin } \E$ 
and
$\E \eqncp \E_1$
imply
$\F \clp^{\CP}_{\min , \bin } \E_1 .$ 
\item
$\F \clp^{\Sch}_{\alg , \lnor } \E$
and
$\E \eqncp \E_1$
imply
$\F \clp^{\Sch}_{\alg, \lnor} \E_1 .$
\end{enumerate}
\end{lemma}

\begin{proof}
\begin{enumerate}
\item
Let $\Lambda \in \cpchset{\M \mintensor \N}{\M}$
be a binormal channel 
that extracts $\F$ without disturbing $\E$
and let $\Gamma \in \ncpchset{\N_1}{\N}$
be a normal channel satisfying
$\pstho = \psth \circ \Gamma$
for all $\thin .$
From Proposition~\ref{prop:cp}
the tensor product map
$\id_{\M} \mintensor \Gamma \in 
\cpchset{\M \mintensor \N_1}{\M \mintensor \N} $
is well-defined.
Then 
$\Lambda_1 
:= \Lambda 
\circ 
(\id_{\M} \mintensor \Gamma)
\in \cpchset{\M \mintensor \N_1}{\M}$
is a binormal channel 
that extracts $\F_1 $ without disturbing $\E ,$
which proves $\F_1 \clp^{\CP}_{\min , \bin} \E .$
\item
Let $\Lambda \in \cpchset{\M \mintensor \N }{\M}$
be a binormal channel that extracts $\F$ without disturbing 
$\E $
and let $\Phi \in \ncpchset{\M_1}{\M}$
and $\Psi \in \ncpchset{\M}{\M_1}$
be normal channels satisfying
\begin{gather}
	\phtho = \phth \circ \Phi ,
	\label{eq:co1}
	\\
	\phth = \phtho \circ \Psi
	\label{eq:co2}
\end{gather}
for all $\thin .$
From Proposition~\ref{prop:cp},
the tensor product channel
$\Phi \mintensor \id_\N
\in 
\cpchset{\M_1 \mintensor \N}{\M \mintensor \N}
$
is well-defined.
We define the channel
$
\Lambda^\prime
:=
\Psi \circ \Lambda \circ (\Phi \mintensor \id_\N)
\in
\cpchset{\M_1 \mintensor \N}{\M_1} .
$
Then for each $\thin ,$ each $A\in \M_1,$
and each $B \in \N ,$
we have
\begin{align*}
	\Lambda^\prime (A \otimes B)
	&=
	\Psi \circ \Lambda (\Phi (A) \otimes B) ,
	\\
	\phtho \circ \Lambda^\prime
	(A \otimes \unit_\N)
	&=
	\phtho \circ \Psi \circ \Lambda
	(\Phi (A) \otimes \unit_\N)
	\\
	&=
	\phth \circ \Lambda 
	(\Phi (A) \otimes \unit_\N)
	\\
	&=
	\phth \circ \Phi (A)
	\\
	&=
	\phtho (A) , \\
	\phtho \circ \Lambda^\prime
	(\unit_{\M_1} \otimes B)
	&=
	\phtho \circ \Psi \circ \Lambda
	(\Phi (\unit_{\M_1}) \otimes B)
	\\
	&=
	\phth
	\circ 
	\Lambda
	(\unit_\M \otimes B)
	\\
	&=
	\psth (B).
\end{align*}
Therefore $\Lambda^\prime$
is a binormal CP channel that extracts
$\F $ without disturbing $\E_1 ,$
which proves 
$\F \clp^{\CP}_{\min , \bin} \E .$
\item
Let $\Lambda \in \schset{\M \atensor \N}{\M}$
be a left-normal Schwarz channel that extracts $\F$ without disturbing $\E$
and let $\Phi \in \ncpchset{\M_1}{\M}$
and $\Psi \in \ncpchset{\M}{\M_1}$
be normal channels satisfying
\eqref{eq:co1} and \eqref{eq:co2}.
From Lemma~\ref{lemm:schwarz}
we have 
$\Lambda \circ (\Phi \atensor \id_\N)
\in
\schset{\M_1 \atensor \N}{\M} ,$
and hence
$
\Lambda^\prime
:=
\Psi \circ \Lambda \circ (\Phi \atensor \id_\N)
\in
\schset{\M_1 \atensor \N}{\M_1} .
$
Then we can show that 
$\Lambda^\prime$
is a left-normal Schwarz channel that
extracts $\F$ without disturbing $\E_1$
in the same way as in the proof of claim~2.
Thus $\F \clp^{\Sch}_{\alg ,  \lnor} \E_1 . $
\qed
\end{enumerate}
\end{proof}

\begin{lemma}
\label{lemm:abelian}
Let $\E = (\M , \Theta , (\phth)_{\thin})$ be a statistical experiment.
Define 
\[\F := (\zent (\M) , \Theta , (\psth)_{\thin}) ,\]
where $\psth$ is the restriction of $\phth$ to the center
$\zent (\M) .$
Then $\F \clp^{\CP}_{\min, \bin} \E .$
\end{lemma}
\begin{proof}
Since $\M$ and $\zent (\M)$ commute,
there exists a representation
$
\pi \colon \M \maxtensor \zent (\M)
=
\M \mintensor \zent (\M)
\to
\M
$
such that
$\pi (A \otimes Z) = AZ$
($A \in \M , Z \in \zent (\M)$).
Then for each $A \in \M,$
each $Z \in \zent (\M) ,$
and each $\thin ,$
we have
\begin{gather*}
	\phth \circ \pi
	(A \otimes \unit_{\M})
	=\phth (A) ,
	\\
	\phth \circ \pi
	(\unit_\M \otimes Z)
	= 
	\phth (Z)
	=
	\psth (Z).
\end{gather*}
Therefore 
$\pi \in \cpchset{\M \mintensor \zent (\M)}{\M}$
is a binormal channel 
that extracts $\F$
without disturbing $\E ,$
which proves 
$\F \clp^{\CP}_{\min, \bin} \E .$
\qed
\end{proof}

\noindent
\textit{Proof of Theorem~\ref{theo:main}.}
Let $\M \gtensor \N$ be an arbitrary \cstar-tensor product of $\M$ and $\N ,$
let $\E_0 =\seEz$ be the minimal sufficient statistical experiment
satisfying $\E_0 \eqncp \E ,$
and let $\E_\cl = (\zent (\M_0) , \Theta , (\phth^\cl)_{\thin})$ be the 
classical part of $\E .$
From the definitions of the accessibility relations $\clp^X_{y, z} ,$
the following implications immediately follow:
\begin{gather*}
	\F \clp^{\CP}_{y, z}\E
	\implies 
	\F \clp^{\Sch}_{y,z} \E,
	\\
	\F \clp^X_{\gamma , z } \E
	\implies
	\F \clp^X_{\alg , z} \E ,
	\\
	\F \clp^X_{y , \bin } \E
	\implies
	\F \clp^X_{y , \lnor} \E
	\implies
	\F \clp^X_{y } \E  .
\end{gather*}
Hence
\[
	\F \clp^{\CP}_{\gamma , \bin} \E
	\implies
	\F \clp^{X}_{y,z} \E
	\implies
	\F \clp^{\Sch}_{\alg} \E 
\]
if $y = \gamma $ or $y = \alg.$
Thus it is sufficient to establish the implications 
(i)$\implies$(ii)$\implies$(iii)$\implies$(iv)$\implies$(v)
for the following conditions:
\begin{enumerate}[(i)]
\item
$\F \clp^{\Sch}_\alg \E ;$
\item
$\F \clp^{\Sch}_{\alg , \lnor} \E ;$
\item
$\F \concp \E_\cl ;$
\item
$\F \clp^{\CP}_{\min , \bin} \E ;$
\item
$\F \clp^{\CP}_{\gamma , \bin} \E .$
\end{enumerate}

(i)$\implies$(ii).
Assume~(i) and take a channel
$\Lambda \in \schset{\M \atensor \N}{\M}$
that extracts $\F$ without disturbing $\E .$
Then from Lemma~\ref{lemm:lnor}, 
$\Lambda$ extends to a left-normal channel
$\widetilde{\Lambda}\in \schset{\M^{\ast\ast}\atensor \N}{\M^{\ast\ast}} .$
If we write the normal extension of $\E $ as 
$\overline{\E} = (\M^{\ast\ast} , \Theta , (\overline{\phth})_{\thin}) ,$
then for each $A \in \M ,$ 
each $B \in \N ,$ 
and each $\thin ,$
we have
\begin{gather}
	\overline{\phth} \circ \widetilde{\Lambda}
	(A \otimes \unit_\N)
	=
	\phth \circ \Lambda (A \otimes \unit_\N)
	=
	\phth (A) 
	=\overline{\phth} (A) ,
	\label{eq:onM} 
	\\
	\overline{\phth} \circ \widetilde{\Lambda}
	(\unit_{\M^{\ast\ast}} \otimes B)
	=
	\phth \circ \Lambda (\unit_\M \otimes B)
	=
	\psth (B) .
	\notag
\end{gather}
Since $\M$ is ultraweakly dense in $\M^{\ast \ast} $
and $\widetilde{\Lambda}$ is left-normal,
\eqref{eq:onM} implies 
\[
	\overline{\phth} \circ \widetilde{\Lambda}
	(A^{\prime\prime} \otimes \unit_\N)
	=
	\overline{\phth} 
	(A^{\prime\prime})
\]
for all $A^{\prime\prime} \in \M^{\ast\ast} .$
Therefore we obtain $\F \clp^{\Sch}_{\alg, \lnor} \overline{\E} .$
From Lemmas~\ref{lemm:nex} and \ref{lemm:coarse},
this implies $\F \clp^{\Sch}_{\alg, \lnor} \E .$

(ii)$\implies$(iii).
Assume (ii).
From Lemma~\ref{lemm:coarse} we have 
$\F \clp^{\Sch}_{\alg , \lnor} \E_0$
and hence we can take a left-normal channel
$\Lambda \in \schset{\M_0 \atensor \N}{\M_0}$
that extracts $\F$ without disturbing $\E_0 .$
We define Schwarz channels $\Lambda_L \in \nsch{\M_0}$
and 
$\Lambda_R \in \schset{\N}{\M_0}$
by
$
	\Lambda_L (A)
	:=
	\Lambda (A \otimes \unit_\N)
$
and
$\Lambda_R (B) := \Lambda (\unit_{\M_0} \otimes B)$
$(A \in \M_0, B \in \N) .$
Since we have
$
	\phthz \circ \Lambda_L
	=
	\phthz
$
for all $\thin ,$
the minimal sufficiency of $\E_0$ implies
$\Lambda_L = \id_{\M_0} ,$
i.e.\
$
	\Lambda (A \otimes \unit_\N)
	=
	A
$
for all $A \in \M_0 .$
Thus for each $A \in \M_0$ we have
\begin{gather*}
	\Lambda ( (A \otimes \unit_\N)^\ast (A \otimes \unit_\N))
	= A^\ast A
	= \Lambda  (A \otimes \unit_\N)^\ast \Lambda (A \otimes \unit_\N) ,
	\\
	\Lambda ( (A \otimes \unit_\N) (A \otimes \unit_\N)^\ast)
	= A A^\ast
	= \Lambda  (A \otimes \unit_\N) \Lambda (A \otimes \unit_\N)^\ast .
\end{gather*}
Therefore $\M_0 \otimes \unit_\N$ is contained in 
the multiplicative domain of $\Lambda .$
Thus, from Lemma~\ref{lemm:mdomain}, for each $A \in \M_0$ and 
each $B \in \N ,$ we have
\begin{align*}
	A \Lambda_R(B)
	&=
	\Lambda (A \otimes \unit_\N)
	\Lambda (\unit_{\M_0} \otimes B)
	\\
	&=
	\Lambda (A \otimes B)
	\\
	&=
	\Lambda (\unit_{\M_0} \otimes B)
	\Lambda (A \otimes \unit_\N)
	\\
	&=
	\Lambda_R (B) A ,
\end{align*}
which implies
$\Lambda_R (\N) \subseteq \zent (\M_0) .$
Since $\zent (\M_0)$ is commutative, 
$\Lambda_R$ is CP.
Furthermore, we have
$
\psth
= \phthz \circ \Lambda_R
= \phth^\cl \circ \Lambda_R 
$
for all $\thin ,$
which implies $\F \cocp \E_\cl .$
Thus from Lemma~\ref{lemm:coeq} we obtain
$\F \concp \E_\cl .$

(iii)$\implies$(iv) 
follows from Lemmas~\ref{lemm:coarse} and \ref{lemm:abelian}.

(iv)$\implies$(v).
Assume (iv) and let 
$\Lambda \in \cpchset{\M \mintensor \N}{\M}$
be a binormal CP channel that extracts $\F$ without disturbing $\E .$ 
From the minimality of the norm $\minnorm{\cdot}$ on $\M \atensor \N ,$
there exists a representation 
$\pi_{\min} \colon \M \gtensor \N \to \M \mintensor \N$
such that $\pi_{\min} (A \otimes B) = A \otimes B$
$(A \in \M , B \in \N).$
We define a channel 
$\Lambda_\gamma \in \cpchset{\M \gtensor \N}{\M}$
by $\Lambda_\gamma := \Lambda \circ \pi_{\min}.$
Then $\Lambda_\gamma$ is a binormal channel 
that extracts $\F$ without disturbing $\E ,$
which implies $\F \clp^{\CP}_{\gamma , \bin} \E .$
\qed

\begin{remark}
\label{rem:lindblad}
The proof of (ii)$\implies$(iii) in Theorem~\ref{theo:main}
is analogous to the proof of Lindblad's ``general no-cloning theorem''
(\cite{Lindblad1999}, Theorem~1),
in which the set of invariant states of a given broadcasting channel is considered.
\end{remark}

Since all the relations $\F \clp^X_{y,Z} \E$ in Definition~\ref{defi:acrelation} coincide,
from now on we adopt the simpler notation $\F \clp \E $ instead of $\F \clp^X_{y,Z} \E .$

\subsection{No-broadcasting} 
\label{subsec:nb}
A statistical experiment $\E = \seE$
is called \emph{broadcastable in the sense of algebraic tensor product}
if there exists a Schwarz channel
$\Lambda \in \schset{\M \atensor \M}{\M}$ 
such that
\[
	\phth (A)
	=
	\phth \circ \Lambda (A \otimes \unit_\M)
	=
	\phth \circ \Lambda ( \unit_\M   \otimes A ) 
\]
for all $\thin $
and all $A \in \M .$
From the definition, we can easily see that 
$\E$ is broadcastable in the sense of algebraic tensor product
if and only if $\E \clp \E .$
Therefore from Theorem~\ref{theo:main}
we immediately obtain

\begin{corollary}[No-broadcasting theorem]
\label{coro:nb}
A statistical experiment $\E $
is broadcastable in the sense of algebraic tensor product 
if and only if
$\E$ is normally CP equivalent to a classical statistical experiment.
\end{corollary}

\begin{remark}
\label{rem:broadcast}
In \cite{Kaniowski2015}, the broadcastability of normal states 
was considered by identifying the outcome composite system
with the normal tensor product $\M \ntensor \M ,$
which is more restrictive than our broadcastability condition here.
Indeed, a statistical experiment $\E$ is broadcastable in the sense of 
\cite{Kaniowski2015} if and only if $\E$ is normally CP equivalent 
to a classical statistical experiment 
with an \emph{atomic} outcome von Neumann algebra,
and hence any minimal sufficient statistical experiment with 
a non-atomic commutative outcome algebra is not broadcastable in this sense.
For the difference of the \cstar- and normal tensor products 
from the view point of broadcasting,
see also Section~7 of \cite{kuramochi2017incomp}.
\end{remark}

\subsection{Density operators}
\label{subsec:KI} 
Let us consider a statistical experiment 
$\E = (\LH , \Theta , (\phth)_{\thin})$
for a Hilbert space $\cH .$
As usual we regard each normal state $\phth$
as a density operator $\rth$ on $\cH$
satisfying
$\phth (A) = \tr (\rth A)$
$(A \in \LH) .$ 
By restricting the outcome Hilbert space $\cH$ if necessary, 
we assume that $\E$ is faithful.
Let $\M_0 \subseteq \LH$ be the minimal sufficient subalgebra 
with respect to $(\phth)_{\thin}$
and let $\condi $ be the normal conditional expectation from $\LH$
onto $\M_0$ satisfying
$\phth \circ \condi = \phth$ for all $\thin .$
Since $(\phth)_{\thin}$
is faithful, so is $\condi .$
Hence,
according to \cite{takesakivol1} (Chapter~V, Section~2, Exercise~8)
$\M_0$ is atomic, i.e.,
we have the following decompositions:
\begin{gather}
	\cH
	=
	\bigoplus_{i \in I}
	\cH_i 
	\otimes
	\cK_i
	,
	\label{eq:KIH}
	\\
	\M_0
	=
	\bigoplus_{i \in I}
	\calL (\cH_i )
	\ntensor 
	\cmplx \unit_{\cK_i} ,
	\notag
\end{gather}
where $\cH_i $ and $\cK_i$ are Hilbert spaces
and $\unit_\cK$ denotes the identity operator on a Hilbert space $\cK .$
(We can also show this from \cite{10.2307/24491050}).
Then, as in the finite-dimensional case
(\cite{Lindblad1999}, Section~4; \cite{Hayden2004}, Appendix~A),
we can show the following decompositions for $\condi$ and $\rth :$
\begin{gather}
	\condi (A)
	=
	\bigoplus_{ i \in I}
	\tr_{\cK_i} (
	P_i A P_i
	(\unit_{\cH_i} \otimes \sigma_i)
	)
	\otimes \unit_{\cK_i} 
	\quad
	(A \in \LH) ,
	\notag
	\\
	\rth
	=
	\bigoplus_{i \in I}
	q_\theta (i)
	\rho_{i,\theta }
	\otimes 
	\sigma_i,
	\label{eq:KIrho}
\end{gather}
where $P_i$ is the orthogonal projection onto $\cH_i \otimes \cK_i ,$
$\tr_{\cK_i} (\cdot)$ is the partial trace over $\cK_i ,$
$q_\theta (i)$ is a probability distribution over $i \in I$
for each $\thin ,$
and $\rho_{i ,\theta}$ and $\sigma_i$
are density operators on $\cH_i$
and $\cK_i, $
respectively. 
If $\cH$ is finite-dimensional,
the decomposition given by \eqref{eq:KIH} and \eqref{eq:KIrho}
coincides with the maximal decomposition 
obtained by Koashi and Imoto~\cite{PhysRevA.66.022318}.

Since we have
$
\zent (\M_0) =
\bigoplus_{i\in I}
\cmplx P_i ,
$
the classical part $\E_\cl$ is normally isomorphic to 
$
	(
	\ell^\infty (I) , \Theta ,  (Q_\theta )_{\thin} 
	) ,
$
where $\ell^\infty (I)$ is the set of bounded complex-valued functions on $I $
and $Q_\theta\in \Ss (\ell^\infty (I))$
is
given by
$Q_\theta (f) = \sum_{i \in I} q_\theta (i) f(i)$
$(f \in \ell^\infty (I)) .$ 
In this sense, the classical part $\E_\cl$ corresponds to 
the probability distributions
$q_\theta (i)$ appearing in the maximal decomposition \eqref{eq:KIrho}.
We can extract, in this case, the information $\E_\cl$ without disturbing $\E$
by performing the discrete projective measurement corresponding to
$(P_i)_{i \in I} .$

We remark that
the infinite-dimensional version of the Koashi-Imoto decomposition was obtained in
\cite{jencovapetz2006}
for separable $\cH$
by using the modular theory in operator algebras.

%

\section{Direct product of statistical experiments}
\label{sec:direct}
In this section we consider the direct product of statistical experiments 
and its classical part.

Let $\E = \seE $ and $\F = \seFx$ be statistical experiments.
We define the direct product $\E \otimes \F$ by
\[
	\E \otimes \F
	:=
	(\M \ntensor \N , \Theta \times \Xi , (\phth \ntensor \psi_\xi)_{(\theta , \xi) \in \Theta \times \Xi} ) .
\]
Operationally, 
$\E \otimes \F$
corresponds to the juxtaposition of 
two partially known systems corresponding to $\E$ and $\F .$

We first show that the direct product of 
minimal sufficient statistical experiments is also minimal sufficient.

\begin{theorem}
\label{theo:directms}
Let $\E = \seE$
and $\F = \seFx$
be statistical experiments.
Then the following conditions are equivalent:
\begin{enumerate}[(i)]
\item
$\E \otimes \F$ is minimal sufficient;
\item
$\E $ and $ \F$ are minimal sufficient.
\end{enumerate}
\end{theorem}

\begin{proof}
(i)$\implies$(ii).
Assume (i).
Let $\Phi \in \ncpch{\M}$ and $\Psi \in \ncpch{\N}$ be normal channels 
satisfying 
$\phth \circ \Phi = \phth$
$(\thin)$
and
$\psi_\xi \circ \Psi = \psi_\xi$
$(\xi \in \Xi) .$
Then we have
\[
(\phth \ntensor \psi_\xi) \circ (\Phi \ntensor \Psi)
=
(\phth \circ \Phi) \ntensor (\psi_\xi \circ \Psi)
=
\phth \ntensor \psi_\xi
\]
for each $( \theta , \xi) \in \Theta \times \Xi .$ 
Thus from the minimal sufficiency of $\E \otimes \F,$
we have $\Phi \ntensor \Psi = \id_{\M \ntensor \N } ,$
which implies 
$\Phi = \id_{\M}$ and 
$\Psi = \id_\N .$
Therefore $\E$ and $\F$ are minimal sufficient.

(ii)$\implies$(i).
Assume (ii).
Since we have 
$\s (\phth \ntensor \psi_\xi) = \s (\phth) \otimes \s (\psi_\xi) $
(\cite{takesakivol1}, Corollary~IV.5.12),
the family $(\phth \ntensor \psi_\xi)_{(\theta , \xi) \in \Theta \times \Xi}$
is faithful on $\M \ntensor \N .$
Therefore, by applying the mean ergodic theorem for von Neumann algebras,
there exists a minimal sufficient subalgebra $\M_0 \subseteq \M \ntensor \N$
with respect to $(\phth \ntensor \psx)_{(\theta , \xi) \in \Theta \times \Xi  }$
and a normal conditional expectation 
$\condi$ from $\M \ntensor \N $ onto $\M_0$ such that
$
	( \phth \ntensor \psx ) \circ \condi
	=
	\phth \ntensor \psx 
$
$
	(\thin , \xin) .
$
We will show $\M \ntensor \N \subseteq \M_0 ,$
which implies the condition~(i).

First, we fix arbitrary $\xin$
and write $\s_\xi := \s (\psx) $
and $\ts_\xi := \unit_\M \otimes \s (\psx) .$
Consider the following faithful statistical experiment:
\[
	\widetilde{\E}_\xi
	:=
	(
	\M \ntensor (\s_\xi \N \s_\xi) ,
	\Theta ,
	(\phth \ntensor \psx)_{\thin}
	) .
\]
Define a map
$
\condi_\xi 
\colon 
\M \ntensor (\s_\xi \N \s_\xi)  
\to
\M \ntensor \cmplx \s_\xi
$
by the normal extension of
$
	\condi_\xi
	(A \otimes B)
	=
	A \otimes \psx (B)  \s_\xi
$
$(A \in \M , B \in \s_\xi \N \s_\xi) .$
Then $\condi_\xi$ is a normal conditional expectation from 
$\M \ntensor (\s_\xi \N \s_\xi)  $
onto
$\M \ntensor \cmplx \s_\xi$
and satisfies
\[
	(\phth \ntensor \psx)
	\circ
	\condi_\xi
	=
	\phth \ntensor \psx
\]
for all $\thin .$
Since the statistical experiment
\[
	(
	\M \ntensor \cmplx \s_\xi,
	\Theta ,
	(\phth \ntensor \psx)_{\thin}
	)
\]
is normally isomorphic to $\E ,$
the assumption of the minimal sufficiency of $\E$ implies that
$\M \ntensor \cmplx \s_\xi$ is the minimal sufficient subalgebra of
$\M \ntensor ( \s_\xi \N \s_\xi )$
with respect to the family
$(\phth \ntensor \psx)_{\thin} $
(\cite{kuramochi2017minimal}, Theorem~3).
Therefore for a normal channel $\Gamma \in \ncpch{\M \ntensor (\s_\xi \N \s_\xi)} ,$
$(\phth \ntensor\psx ) \circ \Gamma (X) = (\phth \ntensor \psx) (X)  $
$(X \in \M \ntensor (\s_\xi \N \s_\xi)    , \,  \thin   )$
implies
$\Gamma (A \otimes \s_\xi)
=
A \otimes \s_\xi $
$( A \in \M ) .$

Since $\s (\phth) \otimes \s (\psx) = 
\s (\phth \ntensor \psx) \in \M_0$
for each $(\theta , \xi) \in \Theta \times \Xi$
(\cite{thomsen1985invariant}, Lemma~1),
we have 
$\ts_\xi = \bigvee_{\thin} \s (\phth) \otimes \s (\psx) \in \M_0 .$
Therefore we obtain
\begin{gather*}
	\condi (\ts_\xi) = \ts_\xi ,
	\\
	\condi (X)
	=
	\condi (
	\ts_\xi
	X
	\ts_\xi
	)
	=
	\ts_\xi 
	\condi (X)
	\ts_\xi 
	\in
	\M \ntensor (\s_\xi \N \s_\xi)
	\quad
	(
	X \in \M \ntensor (\s_\xi \N \s_\xi)
	) .
\end{gather*}
Thus if we define
$\Lambda_\xi$ by the restriction of $\condi$ to
$\M \ntensor (\s_\xi \N \s_\xi) ,$
$\Lambda_\xi$
is a normal channel on $\M \ntensor (\s_\xi \N \s_\xi) .$
Then
$
(\phth \ntensor \psx) \circ \Lambda_\xi (X)
=
(\phth \ntensor \psx) \circ \condi (X)
=
(\phth \ntensor \psx )(X)
$
for each $\thin $ and each $X \in \M \ntensor (\s_\xi \N \s_\xi)  .$
Therefore for each $A \in \M$ we have
\[
	\M_0 \ni
	\condi (A \otimes \s_\xi)
	=
	\Lambda_\xi 
	(A \otimes \s_\xi)
	=
	A \otimes \s_\xi .
\]
Hence we obtain
$\M \ntensor \cmplx \s_\xi \subseteq \M_0 $
for each $\xin .$
Thus for each projection $P \in \M $
we have
\[
	P \otimes \unit_\N
	=
	\bigvee_{\xin}
	P \otimes \s_\xi
	\in
	\M_0 ,
\]
which implies
$\M \otimes \unit_\N \subseteq \M_0 .$
We can also show 
$\unit_\M \otimes \N \subseteq \M_0 $
analogously.
Thus we finally obtain 
$\M \ntensor \N \subseteq \M_0,$
which completes the proof.
\qed
\end{proof}

\begin{remark}
\label{rem:direct}
A corresponding result of Theorem~\ref{theo:directms} 
for dominated families of probability measures
was obtained in~\cite{10.2307/24307476}
and for families of finite-dimensional density operators
in~\cite{PhysRevA.66.022318} (Theorem~4).
\end{remark}

By using Theorem~\ref{theo:directms},
we show that the classical part of a direct product of statistical experiments
is the direct product of the classical parts of the statistical experiments:

\begin{theorem}
\label{theo:directcl}
Let $\E = \seE$
and $\F = \seFx$
be statistical experiments.
Then
$(\E \otimes \F)_\cl \cong \E_\cl \otimes \F_\cl .$
\end{theorem}
\begin{proof}
Let 
$\E_0 = \seEz$ and 
$\F_0 = \seFxz$ be the minimal sufficient statistical 
experiments satisfying 
$\E \eqncp \E_0$
and
$\F \eqncp \F_0 ,$
and let 
$\E_\cl = (\zent (\M_0 ) , \Theta , (\vph^{\cl}_\theta)_{\thin})$
and 
$\F_\cl = (\zent (\N_0 ) , \Xi , (\psi^{\cl}_\xi)_{\xi \in \Xi})$
be the classical parts of $\E$ and $\F ,$ respectively.
Since we have $\E \otimes \F \eqncp \E_0 \otimes \F_0 ,$
Theorem~\ref{theo:directms} implies that 
$\E_0 \otimes \F_0$ is the minimal sufficient statistical experiment
normally CP equivalent to $\E \otimes \F .$
Therefore
\[
	(\E \otimes \F)_\cl
	\cong
	\left(
	\zent (\M_0 \ntensor \N_0) , \,
	\Theta \times \Xi , \,
	\left(
	\lambda_{\theta , \xi}
	\right)_{(\theta , \xi) \in \Theta \times \Xi}
	\right) ,
\]
where $\lambda_{\theta ,\xi}$ is the restriction of 
$\phthz \ntensor \psi^{(0)}_\xi$ to the center 
$\zent (\M_0 \ntensor \N_0 ) .$
On the other hand, it is known that
$\zent (\M_0 \ntensor \N_0) = \zent (\M_0)\ntensor \zent (\N_0)$
(\cite{takesakivol1}, Corollary~IV.5.11).
Thus $\lambda_{\theta ,\xi}$ coincides with 
$\vph^{\cl}_\theta \ntensor \psi^{\cl}_\xi .$
Therefore 
\[
	(\E \otimes \F)_{\cl}
	\cong
	\left(
	\zent (\M_0) \ntensor \zent( \N_0) , \,
	\Theta \times \Xi , \,
	\left(
	\vph^{\cl}_\theta \ntensor \psi^{\cl}_\xi 
	\right)_{(\theta , \xi) \in \Theta \times \Xi}
	\right) 
	=
	\E_{\cl} \otimes \F_{\cl} ,
\]
which completes the proof.
\qed
\end{proof}

\section{Concluding remark}\label{sec:conclusion}
In this paper we have given an operational meaning
of the classical part $\E_\cl$ of a statistical experiment $\E .$
Namely, $\E_\cl$ is the accessible part of $\E$ when we do not disturb $\E .$
In the formulation of the accessibility relations, 
we have identified 
the outcome composite system with 
the algebraic or a \cstar-tensor product of the outcome algebras.
In the von Neumann algebra setup considered in this paper,
it is also possible to identify the outcome composite system with 
the \emph{normal} tensor product.
The accessibility relation defined by the normal tensor product
is more restrictive than the accessibility relation $\clp $ considered in this paper,
similarly as the notion of broadcastability presented in 
\cite{Kaniowski2015} is more restrictive than that defined here.
Then it is natural to ask how these accessibility relations differ,
mathematically and operationally.
We leave this as an open question.

\appendix
\section{Proofs of Lemmas~\ref{lemm:nex} and \ref{lemm:coeq}}
\label{sec:appendix}
In this appendix, we give direct proofs of Lemmas~\ref{lemm:nex} and \ref{lemm:coeq}.

\noindent
\textit{Proof of lemma~\ref{lemm:nex}.}
From the universality of $\M^{\ast\ast} ,$
there exists a normal representation
$\widetilde{\pi} \colon \M^{\ast \ast} \to \M$
satisfying
$\widetilde{\pi} \circ \pi_\M = \id_\M .$
Then for each $A \in \M$ and each $\thin$ we have 
\[
	\overline{\phth} (\pi_\M (A))
	=
	\phth (A)
	=
	\phth \circ \widetilde{\pi} (\pi_\M (A)) .
\]
Since $\pi_\M (\M)$ is ultraweakly dense in $\M^{\ast\ast}$
and $\overline{\phth}$ and $\phth \circ \widetilde{\pi}$ are normal,
we obtain $\overline{\phth} = \phth \circ \widetilde{\pi} $
for each $\thin .$
Hence $\overline{\E} \concp \E .$
On the other hand, 
since the kernel $\ker \widetilde{\pi}$ is a
self-adjoint, ultraweakly closed,
two-sided ideal on $\M^{\ast\ast} ,$
there exists a central projection
$\widetilde{P} \in \zent (\M^{\ast\ast})$
such that 
$\ker \widetilde{\pi} = \widetilde{P}^\perp \M^{\ast\ast} , $
where
$\widetilde{P}^\perp := \unit_{\M^{\ast\ast}} - \widetilde{P} .$
Then the restriction 
$\widetilde{\pi} \rvert_{\widetilde{P} \M^{\ast\ast}}$
of $\widetilde{\pi}$ to the 
(possibly non-unital) von Neumann subalgebra
$\widetilde{P}\M^{\ast\ast} $ is a normal isomorphism
from 
$\widetilde{P}\M^{\ast\ast} $
onto
$\widetilde{\pi} (\M^{\ast \ast}) = \M.$ 
We write the inverse of 
$\widetilde{\pi} \rvert_{\widetilde{P} \M^{\ast\ast}}$
as $\rho \colon \M \to \widetilde{P} \M^{\ast\ast} .$
We define a normal CP channel 
$\Phi \in \ncpchset{\M}{\M^{\ast\ast}}$
by
$
\Phi (A) := \rho (A) + \phi_0 (A) \widetilde{P}^\perp 
$
$(A \in \M) ,$
where $\phi_0$ is a fixed normal state on $\M .$
Then for each $A \in \M $ and each $\thin ,$ we have
\[
	\overline{\phth} \circ \Phi (A)
	=
	\phth \circ \widetilde{\pi} \circ \Phi (A)
	=
	\phth \left(
	\widetilde{\pi} \circ \rho (A)
	+
	\phi_0 (A) \widetilde{\pi} (\widetilde{P}^\perp)
	\right)
	=
	\phth (A) ,
\]
which implies 
$\E \concp \overline{\E} .$
Thus we have shown 
$\E \eqncp \overline{\E} .$
\qed

\noindent
\textit{Proof of Lemma~\ref{lemm:coeq}.}
We have only to prove the implication
$\E \cocp \F \implies \E \concp \F .$
Assume $\E \cocp \F$ and take a channel 
$\Phi \in \cpchset{\M}{\N}$ such that
$\phth = \psth \circ \Phi$
for all $\thin .$
Let $\Phi^{\ast\ast} \colon \M^{\ast\ast} \to \N^{\ast \ast}$
be the double dual map of $\Phi .$
Then $\Phi^{\ast\ast}$ is an extension of $\Phi$ 
that is continuous in the ultraweak topologies
of $\M^{\ast\ast}$
and 
$\N^{\ast\ast} ,$
respectively.

Now we show that $\Phi^{\ast\ast}$ is CP.
For this, we have to show that 
for each integer $n \geq 1$
and each 
$\{A_i^{\prime\prime}\}_{i=1}^n \subseteq \M^{\ast\ast} ,$
the $n\times n$ matrix
$(\Phi^{\ast \ast} (A_i^{\prime\prime \ast} A_j^{\prime\prime}))_{i,j =1}^n$
is positive.
By the Kaplansky density theorem, we can take nets
$(A_{i\alpha})$ on $\M$ satisfying
$\norm{A_{i\alpha}} \leq \norm{A_i^{\prime\prime}}$
and 
$\pi_\M (A_{i\alpha} ) \xrightarrow{us\ast} A^{\prime\prime}_i $
$(i = 1, \dots , n) ,$
where $\xrightarrow{us\ast}$ denotes the ultrastrong$\ast$
convergence on $\cH_\M .$
Then we have 
$\pi_\M ( A_{i\alpha}^{\ast} A_{j\alpha} ) 
\xrightarrow{us\ast}
A_i^{\prime\prime \ast} A_j^{\prime\prime}
$
and hence
$
\pi_\N \circ
\Phi (A_{i\alpha}^{\ast} A_{j\alpha} ) 
=
\Phi^{\ast\ast} \circ \pi_\M (A_{i\alpha}^{\ast} A_{j\alpha} )
\xrightarrow{uw}
\Phi^{\ast\ast}(
A_i^{\prime\prime \ast} A_j^{\prime\prime}
) ,
$
where $\xrightarrow{uw}$ denotes the ultraweak convergence.
Hence the matrix
$(\Phi^{\ast \ast} (A_i^{\prime\prime \ast} A_j^{\prime\prime}))_{i,j =1}^n$
is an ultraweak limit of the net of positive matrices
$(\pi_\N \circ \Phi(A_{i\alpha}^{\ast} A_{j\alpha}))_{i,j =1}^n ,$
and therefore is also positive.
Thus $\Phi^{\ast\ast}$ is a normal CP channel.

We denote the normal extensions of
$\E$ and $\F$
by
$\overline{\E} = (\M^{\ast\ast} , \Theta , (\overline{\phth})_{\thin})$
and
$\overline{\F} = (\N^{\ast\ast} , \Theta , (\overline{\psth})_{\thin}) , $
respectively.
Then for each $A \in \M $ we have
\[
	\overline{\psth} \circ \Phi^{\ast\ast}
	(\pi_\M (A))
	=
	\overline{\psth}
	(
	\pi_\N (
	\Phi (A)
	)
	)
	=
	\psth \circ \Phi (A)
	= 
	\phth (A)
	=
	\overline{\phth}
	(\pi_\M (A)) .
\]
Since $\pi_\M (\M)$ is ultraweakly dense in
$\M^{\ast\ast} ,$
this implies 
$\overline{\phth} = \overline{\psth} \circ \Phi^{\ast\ast} $
$(\thin) ,$
and hence $\overline{\E} \concp \overline{\F} .$
Thus from Lemma~\ref{lemm:nex},
we obtain $\E \concp \F ,$
which completes the proof.
\qed

\begin{acknowledgements}
The author would like to thank Takayuki Miyadera for 
helpful discussions.
He is also grateful to Erkka Haapasalo 
for helpful comments on the first version of this paper.
\end{acknowledgements}


\end{document}